\documentclass[11pt]{article}%
\usepackage{amsfonts}
\usepackage{amsmath}
\usepackage{amssymb}
\usepackage{graphicx}%
\providecommand{\U}[1]{\protect\rule{.1in}{.1in}}
\newtheorem{theorem}{Theorem}
\newtheorem{acknowledgement}[theorem]{Acknowledgement}

\newtheorem{corollary}[theorem]{Corollary}

\newtheorem{definition}[theorem]{Definition}

\newtheorem{proposition}[theorem]{Proposition}
\newtheorem{remark}[theorem]{Remark}

\newenvironment{proof}[1][Proof]{\noindent\textbf{#1.} }{\ \rule{0.5em}{0.5em}}
\begin{document}

\title{On the Non-Uniqueness of Statistical Ensembles Defining a Density Operator and
a Class of Mixed Quantum States with Integrable Wigner Distribution}
\author{Charlyne de Gosson\thanks{charlyne.degosson@gmail.com}\\University of Vienna
\and Maurice de Gosson\thanks{maurice.de.gosson@univie.ac.at}\\University of Vienna}
\maketitle

\begin{abstract}
It is standard to assume that the Wigner distribution of a mixed quantum state
consisting of square-integrable functions is a quasi-probability distribution,
that is that its integral is one and that the marginal properties are
satisfied. However this is in general not true. We introduce a class of
quantum states for which this property is satisfied, these states are dubbed
\textquotedblleft Feichtinger states\textquotedblright\ because they are
defined in terms of a class of functional spaces (modulation spaces)
introduced in the 1980's by H. Feichtinger. The properties of these states are
studied, which gives us the opportunity to prove an extension to the general
case of a result of Jaynes on the non-uniqueness of the statistical ensemble
generating a density operator.

\end{abstract}

\section{Introduction}

A useful device when dealing with density operators is the covariance matrix,
whose existence is taken for granted in most elementary and advanced texts.
However, a closer look shows that the latter only exists under rather
stringent conditions on the involved quantum states. This question is related
to another more general one. Assume that we are dealing with mixed quantum
state $\{(\psi_{j},\alpha_{j})\}$ with $\psi_{j}\in L^{2}(\mathbb{R}^{n})$,
$||\psi_{j}||=1$ and $\alpha_{j}\geq0$, $\sum_{j}\alpha_{j}=1$ the index $j$
belonging to some countable set. The Wigner distribution of that state, which
is by definition the convex sum of Wigner transforms
\begin{equation}
\rho=\sum_{j}\alpha_{j}W\psi_{j}\label{1}%
\end{equation}
is usually referred to as a \textquotedblleft
quasi-probability\textquotedblright. But such an interpretation makes sense
only if
\begin{equation}
\int_{\mathbb{R}^{2n}}\rho(z)dz=1\label{2}%
\end{equation}
which is the case if $\psi_{j}\in L^{2}(\mathbb{R}^{n})$, but it does not
imply \textit{per se} the absolute integrability of $W\psi_{j}$: there are
examples of square-integrable functions $\psi$ such that $W\psi\in
L^{2}(\mathbb{R}^{n})$ but $W\psi\notin L^{1}(\mathbb{R}^{2n})$ \cite{Gro}. In
\cite{CoTa} Cordero and Tabacco give a simple example in dimension $n=1$: the
function $\psi=\chi_{\lbrack-/2,1/2]}$ (the characteristic function of the
interval $[-/2,1/2]$ belongs to both $L^{1}(\mathbb{R})$ and $L^{2}%
(\mathbb{R})$ but we have $W\psi\notin L^{1}(\mathbb{R}^{2})$. On the other
hand, even when satisfied, condition (\ref{2}) is not sufficient to ensure the
existence of the covariances
\begin{equation}
\sigma_{x_{j}p_{k}}^{2}=\int_{\mathbb{R}^{2n}}x_{j}p_{k}\rho(z)dz\label{3}%
\end{equation}
since these involve the calculation of second moments, and the integrability
of $\rho$ does not guarantee the convergence of such integrals. These
difficulties are usually ignored in the physical literature, or dismissed by
vague assumptions like \textquotedblleft sufficiently fast\textquotedblright%
\ decrease of the Wigner distribution at infinity. 

The aim of this paper is to remedy this vagueness by proposing an adequate
functional framework for a rigorous analysis of mixed states and of the
associated density operator. For this purpose we will use the Feichtinger
modulation spaces $M_{s}^{1}(\mathbb{R}^{n})$, in instead of $L^{2}%
(\mathbb{R}^{n})$ as a \textquotedblleft reservoir\textquotedblright\ for
quantum states. In the simplest case, $s=0$, we get the so-called Feichtinger
algebra $M_{0}^{1}(\mathbb{R}^{n})=S_{0}(\mathbb{R}^{n})$ which can be defined
by
\begin{equation}
\psi\in S_{0}(\mathbb{R}^{n})\Longleftrightarrow\left\{
\begin{array}
[c]{c}%
\psi\in L^{2}(\mathbb{R}^{n})\\
W\psi\in L^{1}(\mathbb{R}^{2n})
\end{array}
\right.  ~. \label{4}%
\end{equation}
We agree that it is not clear at all why the condition above should define a
vector space -- let alone an algebra! -- since the Wigner transform is not
additive; surprisingly enough this is however the case, and one sees that if
we assume that the state $\{(\psi_{j},\alpha_{j})\}$ is such that $\psi_{j}\in
S_{0}(\mathbb{R}^{n})$ for every $j$ then the normalization condition
(\ref{2}) is satisfied. We will also see that the existence of covariances
like (\ref{3}) (and hence the covariance matrix) is guaranteed if we make the
sharper assumption that $\psi_{j}\in M_{s}^{1}(\mathbb{R}^{n})$ for some
$s\geq2$. While the theory of modulation spaces has become a standard tool in
time-frequency and harmonic analysis, it is somewhat less known in quantum
physics. In \cite{golu} we have applied it to deformation quantization; in
\cite{Dias} spectral and regularity results for operators in modulation spaces
are studied. One of the reasons why this theory is less popular in quantum
mechanics might be that the usual treatments are given in terms of short-time
Fourier transforms (also called Gabor transforms) instead of the Wigner
transform as we do here, and this has the tendency to make the theory obscure
for many physicists ignoring the simple relation between Wigner and Gabor
transforms. We therefore speculate that this lack of communication between
both communities is of a pedagogical nature. The use of Wigner transforms in
the theory of modulation spaces actually has many advantages, for instance it
makes the symplectic invariance of these spaces become obvious and thus links
them directly to the Weyl--Wigner--Moyal formalism.

The notation used here is standard, the phase space variable is $z=(x,p)$ with
$x=(x_{1},...,x_{n})\in\mathbb{R}^{n}$ and $p=(p_{1},...,p_{n})\in
(\mathbb{R}^{n})^{\ast}\equiv\mathbb{R}^{n}$. The scalar product on
$L^{2}(\mathbb{R}^{n})$ is
\[
(\psi|\phi)=\int_{\mathbb{R}^{n}}\psi(x)\overline{\phi(x)}dx
\]
($\overline{\phi(x)}$ the complex conjugate of $\phi(x)$) and the associated
norm is denoted by $||\psi||$.

\section{The Modulation Spaces $M_{s}^{1}(\mathbb{R}^{n})$}

We give here a brief review of the main definitions and properties of the
class of modulations spaces we will need. Modulation spaces were introduced by
Feichtinger in the early 1980's \cite{Hans1,Hans2,fe06}. The most complete
treatment can be found in the book \cite{Gro} by Gr\"{o}chenig; also see the
recent review paper \cite{Jakobsen}. In \cite{Birkbis} (Chapters 16 and 17)
modulation spaces are studied from the point of view of the Wigner transform
which we use here.

We will denote by $W(\psi,\phi)$ the cross-Wigner function of
\[
(\psi,\phi)\in L^{2}(\mathbb{R}^{n})\times L^{2}(\mathbb{R}^{n})~;
\]
it is defined by
\begin{equation}
W(\psi,\phi)(z)=\left(  \tfrac{1}{2\pi\hbar}\right)  ^{n}\int_{\mathbb{R}^{n}%
}e^{-\frac{i}{\hbar}p\cdot y}\psi(x+\tfrac{1}{2}y)\overline{\phi(x-\tfrac
{1}{2}y)}dy~. \label{8}%
\end{equation}
When $\psi=\phi$ one obtains the usual Wigner transform
\begin{equation}
W\psi(z)=\left(  \tfrac{1}{2\pi\hbar}\right)  ^{n}\int_{\mathbb{R}^{n}%
}e^{-\frac{i}{\hbar}p\cdot y}\psi(x+\tfrac{1}{2}y)\overline{\psi(x-\tfrac
{1}{2}y)}dy~.
\end{equation}
Recall \cite{Birkbis,WIGNER} that $W(\psi,\phi)$ is a continuous function
belonging to $L^{2}(\mathbb{R}^{2n})$ and that
\begin{equation}
|W(\psi,\phi)(z)|\leq\left(  \tfrac{2}{\pi\hbar}\right)  ^{n}||\psi||~||\phi||
\label{bound}%
\end{equation}
as well as
\begin{equation}
\int W(\psi,\phi)(z)dz=(\psi|\phi)~. \label{Wnorm}%
\end{equation}
Taking $\psi=\phi$ \ it follows that, in particular,%
\begin{equation}
\int W\psi(z)dz=||\psi||~
\end{equation}
hence the integral condition (\ref{2}) holds as soon as $\psi\in
L^{2}(\mathbb{R}^{n})$ is normalized to one.

In what follows $s$ is a non-negative real number: $s\geq0$. We set $z=(x,p)$
and
\begin{equation}
\langle z\rangle=(1+|z|^{2})^{1/2}~.\label{5}%
\end{equation}
The function $z\longmapsto\langle z\rangle$ is the Weyl symbol of the elliptic
pseudodifferential operator $(1-\Delta)^{1/2}$ where $\Delta$ is the Laplacian
in the $z$ variables. We denote by $L_{s}^{1}(\mathbb{R}^{2n})$ the weighted
$L^{1}$-space is defined by%
\begin{equation}
L_{s}^{1}(\mathbb{R}^{2n})=\{\rho:\mathbb{R}^{2n}\longrightarrow
\mathbb{C}:\langle z\rangle^{s}\rho\in L^{1}(\mathbb{R}^{2n})\}~,\label{6}%
\end{equation}
that is, $\rho\in L_{s}^{1}(\mathbb{R}^{2n})$ if and only if we have
\begin{equation}
\Vert\rho\Vert_{L_{s}^{1}}=\int|\rho(z)|\langle z\rangle^{s}dz<\infty
~.\label{6bis}%
\end{equation}
Due to the submultiplicativity of the weights these space are in fact Banach
algebras with respect to convolution (see \cite{rest00}). The same is true for
the spaces $M_{s}^{1}(\mathbb{R}^{n})$.

\begin{definition}
The modulation space $M_{s}^{1}(\mathbb{R}^{n})$ consists of all $\psi\in
L^{2}(\mathbb{R}^{n})$ such that
\begin{equation}
W(\psi,\phi)\in L_{s}^{1}(\mathbb{R}^{2n}) \label{7}%
\end{equation}
for every $\phi\in\mathcal{S}(\mathbb{R}^{n})$ (the Schwartz space of test
functions deceasing rapidly at infinity).
\end{definition}

It turns out that it suffices to check that condition (\ref{7}) holds for
\textit{one} function $\phi\neq0$ for it then holds for \textit{all}; moreover
the mappings $\psi\longmapsto||\psi||_{\phi,M_{s}^{1}}$ defined by
\[
||\psi||_{\phi,M_{s}^{1}}=||W(\psi,\phi)||_{L_{s}^{1}(\mathbb{R}^{2n})}%
=\int_{\mathbb{R}^{2n}}|W(\psi,\phi)(z)|\left\langle z\right\rangle ^{s}dz
\]
form a family of equivalent norms, and the topology on $M_{s}^{1}%
(\mathbb{R}^{n})$ thus defined makes it into a Banach space. We have the chain
of inclusions
\begin{equation}
\mathcal{S}(\mathbb{R}^{n})\subset M_{s}^{1}(\mathbb{R}^{n})\subset
L^{1}(\mathbb{R}^{n})\cap\mathcal{F}(L^{1}(\mathbb{R}^{n}))\subset
L^{1}(\mathbb{R}^{n})\cap C^{0}(\mathbb{R}^{n}) \label{incl}%
\end{equation}
where $\mathcal{S}(\mathbb{R}^{n})$ is the Schwartz space of tests functions
and $\mathcal{F}(L^{1}(\mathbb{R}^{n}))$ is the space of Fourier transforms
$\mathcal{F}\psi$ of the elements $\psi$ of $L^{1}(\mathbb{R}^{n})$. Observe
that
\[
M_{s}^{1}(\mathbb{R}^{n})\subset M_{s^{\prime}}^{1}(\mathbb{R}^{n})\text{
\textit{if and only if} }s\geq s^{\prime}~;
\]
and one proves \cite{Gro} that
\[%
%TCIMACRO{\tbigcap _{s\geq0}}%
%BeginExpansion
{\textstyle\bigcap_{s\geq0}}
%EndExpansion
M_{s}^{1}(\mathbb{R}^{n})=\mathcal{S}(\mathbb{R}^{n})~.
\]

Recall that the metaplectic group $\operatorname*{Mp}(n)$ is the unitary
representation on $L^{2}(\mathbb{R}^{n})$ of the symplectic group.
$\operatorname*{Sp}(n)$. The modulation spaces $M_{s}^{1}(\mathbb{R}^{n})$ are
invariant under the action of metaplectic group: if $\widehat{S}%
\in\operatorname*{Mp}(n)$ and $\psi\in M_{s}^{1}(\mathbb{R}^{n})$ then
$\widehat{S}\psi\in M_{s}^{1}(\mathbb{R}^{n})$. This property actually follows
from the symplectic covariance property
\begin{equation}
\label{covarfei}W(\widehat{S}\psi)(z)=W\psi(S^{-1}z)
\end{equation}
of the Wigner transform, where $S\in\operatorname*{Sp}(n)$ is the projection
of $\widehat{S}\in\operatorname*{Mp}(n)$ (see \cite{Birkbis} for a detailed
study of the metaplectic representation and symplectic covariance).

When $s=0$ we write $M_{0}^{1}(\mathbb{R}^{n})=S_{0}(\mathbb{R}^{n})$ (the
Feichtinger algebra); clearly $M_{s}^{1}(\mathbb{R}^{n})\subset S_{0}%
(\mathbb{R}^{n})$ for all $s\geq0$. In addition to being a vector space,
$S_{0}(\mathbb{R}^{n})$ is a Banach algebra\ for both pointwise multiplication
and convolution. It is actually the smallest Banach algebra invariant under
the action of the metaplectic group $\operatorname*{Mp}(n)$ and phase space
translations. As mentioned in the introduction, the Feichtinger algebra can be
characterized by the condition (\ref{4}): $S_{0}(\mathbb{R}^{n})$ is the
vector space of all $\psi\in L^{2}(\mathbb{R}^{n})$ such that $W\psi\in
L^{1}(\mathbb{R}^{2n})$.

\section{Feichtinger States}

Consider, as in the introduction, a mixed quantum state $\{(\psi_{j}%
,\alpha_{j})\}$ and denote by
\[
\widehat{\rho}=\sum_{j}\alpha_{j}\widehat{\rho}_{j}~\ ,~\rho=\sum_{j}%
\alpha_{j}W\psi_{j}%
\]
the corresponding density operator and its Wigner distribution; $\widehat{\rho
}_{j}$ is the orthogonal projection on the ray $\mathbb{C}\psi_{j}$, that is
\[
\widehat{\rho}_{j}\psi=(\psi|\psi_{j})\psi_{j},\quad\psi\in L^{2}%
(\mathbb{R}^{n}).
\]
Observe that $(2\pi\hbar)^{n}\rho$ is the Weyl symbol of the operator
$\widehat{\rho}$ \cite{Birkbis}.

\begin{definition}
A mixed state $\{(\psi_{j},\alpha_{j})\}$ is a Feichtinger state if and only
we have $\psi_{j}\in M_{s}^{1}(\mathbb{R}^{n})$ for some $s\geq0$.
\end{definition}

Feichtinger states are preserved by the action of the metaplectic group:

\begin{proposition}
Let $\{(\psi_{j},\alpha_{j})\}$ be a Feichtinger state and $\widehat{S}%
\in\operatorname*{Mp}(n)$. Then $\{(\widehat{S}\psi_{j},\alpha_{j})\}$ is also
a a Feichtinger state.
\end{proposition}

\begin{proof}
It is an immediate consequence of metaplectic invariance of the modulation
spaces $M_{s}^{1}(\mathbb{R}^{n})$.
\end{proof}

The Wigner distribution $\rho$ of a Feichtinger state is a \textit{bona-fide} quasi-distribution:

\begin{proposition}
\label{Prop1}Let $\{(\psi_{j},\alpha_{j})\}$ be a Feichtinger state. Then
$\rho\in L^{1}(\mathbb{R}^{2n})$ and the marginal properties%
\begin{align}
\int_{\mathbb{R}^{n}}\rho(z)dp  &  =\sum_{j}\alpha_{j}|\psi_{j}(x)|^{2}%
\label{marg1}\\
\int_{\mathbb{R}^{n}}\rho(z)dx  &  =\sum_{j}\alpha_{j}|F\psi_{j}(x)|^{2}
\label{marg2}%
\end{align}
hold, and we have
\begin{equation}
\int_{\mathbb{R}^{2n}}\rho(z)dz=1~. \label{norm}%
\end{equation}

\end{proposition}

\begin{proof}
Since $M_{s}^{1}(\mathbb{R}^{n})\subset M_{0}^{1}(\mathbb{R}^{n}%
)=S_{0}(\mathbb{R}^{n})$ we automatically have $W\psi_{j}\in L^{1}%
(\mathbb{R}^{2n})$ for every $j$ hence $\rho\in L^{1}(\mathbb{R}^{2n})$. On
the other hand we know that the marginal conditions
\begin{align}
\int_{\mathbb{R}^{n}}W\psi_{j}(z)dp  &  =|\psi_{j}(x)|^{2}\label{margbis1}\\
\int_{\mathbb{R}^{n}}W\psi_{j}(z)dx  &  =|F\psi_{j}(p)|^{2} \label{margbis2}%
\end{align}
hold if both $\psi_{j}$ and $F\psi_{j}$ are integrable \cite{WIGNER}, and this
is precisely the case here in view of the inclusion (\ref{incl}). It follows
that
\[
\int_{\mathbb{R}^{2n}}\rho(z)dz=\sum_{j}\alpha_{j}\int_{\mathbb{R}^{2n}}%
W\psi_{j}(x,p)dpdx=1
\]
since the $\alpha_{j}$ sum up to one and the $\psi_{j}$ are unit vectors
in.$L^{2}(\mathbb{R}^{n})$.
\end{proof}

\begin{remark}
The assumption that $\{(\psi_{j},\alpha_{j})\}$ is a Feichtinger state is
crucial since it ensures us that the marginal properties (\ref{margbis1}) and
(\ref{margbis2}) hold.
\end{remark}

Modulation spaces also allow to define rigorously the covariance matrix
$\Sigma$ of a state. It is the symmetric $2n\times2\times n$ matrix defined
by
\begin{equation}
\Sigma=\int\nolimits_{\mathbb{R}^{2n}}(z-\overline{z})(z-\overline{z})^{T}%
\rho(z)dz \label{vectorform}%
\end{equation}
where $\overline{z}$, the expectation vector, is given by
\begin{equation}
\overline{z}=\int\nolimits_{\mathbb{R}^{2n}}z\rho(z)dz \label{average}%
\end{equation}
(all vectors $z$ are viewed as column matrices in these definitions). Assuming
$\overline{z}=0$ the covariance matrix explicitly given by
\[
\Sigma=%
\begin{pmatrix}
\Sigma_{XX} & \Sigma_{XP}\\
\Sigma_{PX} & \Sigma_{PP}%
\end{pmatrix}
\text{ \ , \ }\Sigma_{PX}=\Sigma_{XP}^{T}%
\]
where $\Sigma_{XP}=(\sigma_{x_{j}p_{k}}^{2})_{1\leq j,k,\leq n}$ with
$\sigma_{x_{j}p_{k}}^{2}$ given by (\ref{3}), and so on.

\begin{proposition}
\label{Prop2}Assume that $\{(\psi_{j},\alpha_{j})\}$ is a Feichtinger state
with $s\geq2$. (i) Then the covariance matrix $\Sigma$ is well-defined; (ii)
the Fourier transform $F\rho$ of the Wigner distribution of $\widehat{\rho}$
is twice continuously differentiable: $F\rho\in C^{2}(\mathbb{R}^{2n})$.
\end{proposition}

\begin{proof}
It suffices to assume that $s=2$; we then have%
\begin{equation}
\int\nolimits_{\mathbb{R}^{2n}}|\rho(z)|(1+|z|^{2})dz<\infty~. \label{cond25}%
\end{equation}
Setting $z_{\alpha}=x_{\alpha}$ if $1\leq\alpha\leq n$ and $z_{\alpha
}=p_{\alpha}$ if $n+1\leq\alpha\leq2n$ we have $\langle z\rangle=(\langle
z_{1}\rangle,...,\langle z_{n}\rangle)$ where $\langle z_{\alpha}\rangle$ is
given by the absolutely convergent integral
\[
z_{\alpha}=\int\nolimits_{\mathbb{R}^{2n}}z_{\alpha}\rho(z)dz
\]
similarly, the integral
\[
\overline{z_{\alpha}z_{\beta}}=\int\nolimits_{\mathbb{R}^{2n}}z_{\alpha
}z_{\beta}\rho(z)dz
\]
is also absolutely convergent in view of the trivial inequalities $|z_{\alpha
}z_{\beta}|\leq1+|z|^{2}$. We have%
\[
F\rho(z)=\left(  \tfrac{1}{2\pi\hbar}\right)  ^{n}\int_{\mathbb{R}^{2n}%
}e^{-\frac{i}{\hbar}z\cdot z^{\prime}}\rho(z^{\prime})dz^{\prime}~;
\]
differentiating twice under the integration sign we get
\[
\partial_{z_{\alpha}}F\rho=-\tfrac{i}{\hbar}F(z_{\alpha}\rho)\text{ \ },\text{
\ }\partial_{z_{\alpha}}\partial_{z_{\beta}}F\rho=\left(  -\tfrac{i}{\hbar
}\right)  ^{2}F(z_{\alpha}z_{\beta}\rho)
\]
hence the estimates%
\begin{align*}
|\partial_{z_{\alpha}}F\rho(z)|  &  \leq\tfrac{1}{\hbar}\left\vert
\int_{\mathbb{R}^{2n}}z_{\alpha}\rho(z)dz\right\vert <\infty\\
|\partial_{z_{\alpha}}\partial_{z_{\beta}}F\rho(z)|  &  \leq\left(  \tfrac
{1}{\hbar}\right)  ^{2}\left\vert \int_{\mathbb{R}^{2n}}z_{\alpha}z_{\beta
}\rho(z)dz\right\vert <\infty~.
\end{align*}

\end{proof}

\section{Independence of the Statistical Ensemble}

Several distinct statistical ensembles can give rise to the same density
matrix. For instance, given an arbitrary mixed state $\{(\psi_{j},\alpha
_{j})\}$ as above, the density matrix $\widehat{\rho}=\sum_{j}\alpha
_{j}\widehat{\rho}_{j}$ where the $\widehat{\rho}_{j}$ are the orthogonal
projection on the rays $\mathbb{C}\psi_{j}$ can be written, using the spectral
decomposition theorem, as $\widehat{\rho}=\sum_{j}\lambda_{j}\widehat{{\rho
_{j}^{\prime}}}$ where the $\lambda_{j}$ are the eigenvalues of $\widehat{\rho
}$ and the $\widehat{{\rho_{j}^{\prime}}}$ is the orthogonal projections on
the rays $\mathbb{C}\phi_{j}$ the $\phi_{j}$ being the eigenvectors
corresponding to the $\lambda_{j}$.

The main result of this section is a generalization to the
infinite-dimensional case of a result due to Jaynes \cite{Jaynes}. Recall that
a partial isometry is an operator whose restriction to the orthogonal
complement of its null-space is an isometry \cite{Halmos}.

\begin{proposition}
\label{Prop3}Let $\{(\psi_{j},\lambda_{j})\}$ be a mixed state and
\ $(\phi_{j})$ and orthonormal basis of $L^{2}(\mathbb{R}^{n})$ and write
$\psi_{j}=\sum_{k}a_{jk}\phi_{k}$. (i) The operator $\widehat{A}$ defined by
\[
\widehat{A}\phi_{j}=\sum_{k}\lambda_{j}^{1/2}a_{jk}\phi_{k}%
\]
is a Hilbert--Schmidt operator and $\widehat{\rho}=\widehat{A}\widehat{A}%
^{\ast}$ is the density matrix of the state $\{(\psi_{j},\lambda_{j})\}$. (ii)
Two mixed states $\{(\psi_{j},\lambda_{j})\}$ and $\{(\psi_{j}^{\prime
},\lambda_{j}^{\prime})\}$ generate the same density matrix $\widehat{\rho
}=\widehat{A}\widehat{A}^{\ast}$ if and only if there exists a partial
isometry $\widehat{U}$ of $L^{2}(\mathbb{R}^{n})$ such that $\widehat{A}%
=\widehat{A^{\prime}}\widehat{U}$ where $\widehat{A^{\prime}}$ is defined in
terms of $\{(\psi_{j}^{\prime},\lambda_{j}^{\prime})\}$.
\end{proposition}

\begin{proof}
\textit{(i)} We have
\[
(\widehat{A}\phi_{j}|\widehat{A}\phi_{j})=\lambda_{j}\sum_{k,\ell}%
a_{jk}\overline{a_{j\ell}}(\phi_{k}|\phi_{\ell})=\lambda_{j}\sum_{k}%
|a_{jk}|^{2}=\lambda_{j}%
\]
since $||\psi_{j}||^{2}=\sum_{k}|a_{jk}|^{2}=1$. It follows that
\begin{equation}
\sum_{j}(\widehat{A}\phi_{j}|\widehat{A}\phi_{j})=\sum_{j}\lambda_{j}=1
\label{agi1}%
\end{equation}
hence $\widehat{A}$ is a Hilbert--Schmidt operator. It follows that
$\widehat{A}^{\ast}$ is also Hilbert--Schmidt, hence $\widehat{A}%
\widehat{A}^{\ast}$ is a trace class operator with unit trace:
\[
\operatorname*{Tr}(\widehat{A}\widehat{A}^{\ast})=\operatorname*{Tr}%
(\widehat{A}^{\ast}\widehat{A})=1
\]
the second equality in view of (\ref{agi1}). Let us show that in fact
$\widehat{\rho}=\widehat{A}\widehat{A}^{\ast}$, that is
\[
\widehat{\rho}\psi=\sum_{k}\lambda_{k}(\psi|\psi_{k})\psi_{k}%
\]
for every $\psi\in L^{2}(\mathbb{R}^{n})$. It is sufficient to show that this
identity holds for the basis vectors $\phi_{j}$, that is%
\begin{equation}
\widehat{\rho}\phi_{j}=\sum_{k}\lambda_{k}(\phi_{j}|\psi_{k})\psi_{k}
\label{A}%
\end{equation}
for every $j$. Using the expansions
\[
\psi_{k}=\sum_{m}a_{km}\phi_{m}=\sum_{\ell}a_{k\ell}\phi_{\ell}%
\]
we can rewrite this identity as%
\[
\widehat{\rho}\phi_{j}=\sum_{k,\ell}\lambda_{k}\overline{a_{kj}}a_{k\ell}%
\phi_{\ell}~.
\]
On the other hand, by definition of $\widehat{A}$ we have
\[
\widehat{A}^{\ast}\phi_{j}=\sum_{k}\lambda_{k}^{1/2}\overline{a_{kj}}\phi_{j}%
\]
hence, by linearity,%
\begin{equation}
\widehat{A}\widehat{A}^{\ast}\phi_{j}=\sum_{k}\lambda_{k}^{1/2}\overline
{a_{kj}}\widehat{A}\phi_{j}=\sum_{k,\ell}\lambda_{k}\overline{a_{kj}}a_{k\ell
}\phi_{\ell} \label{B}%
\end{equation}
which is the same thing as $\widehat{\rho}\phi_{j}$. \textit{(ii) }Let
$\widehat{U}$ be a partial isometry of $L^{2}(\mathbb{R}^{n})$; then
$\widehat{A^{\prime}}\widehat{A^{\prime}}^{\ast}=$ $\widehat{A}\widehat{A}%
^{\ast}=\widehat{\rho}$. Suppose conversely that $\widehat{A^{\prime}%
}\widehat{A^{\prime}}^{\ast}=$ $\widehat{A}\widehat{A}^{\ast}$. A classical
result from the theory of Hilbert spaces (Douglas' lemma \cite{Douglas}) tells
us that there exists a partial isometry $\widehat{U}$ such that
$\widehat{A^{\prime}}=\widehat{A}\widehat{U}$ so we are done..
\end{proof}

\begin{remark}
The result above has been considered and proven by \cite{Bergou,Nielsen,NC} in
the case of quantum states having a finite numbers of elements. Their proofs
do not immediately extend to the infinite dimensional case.
\end{remark}

An immediate consequence is that the property of being a Feichtinger state is
invariant under transformations preserving the density matrix.

\begin{corollary}
\label{Prop4}Let $\{(\psi_{j},\alpha_{j}):j\in J\}$be a Feichtinger state
where $J$ is a finite set of indices. Then every state $\{(\phi_{j},\beta
_{j})$ generating the same density matrix $\widehat{\rho}$ is also a
Feichtinger state. In particular the spectral decomposition $\widehat{\rho
}=\sum_{j\in J}\lambda_{j}\widehat{\rho}_{j}$ consists of orthogonal
projections $\widehat{\rho}_{j}$ on rays $\mathbb{C}\psi_{j}$ where $\psi
_{j}\in M_{s}^{1}(\mathbb{R}^{n})$.
\end{corollary}

\begin{proof}
Assume that $\psi_{j}\in M_{s}^{1}(\mathbb{R}^{n})$ for every $j$. In view of
Proposition \ref{Prop3} there exist finite linear relations $\phi_{k}=\sum
_{j}a_{jk}\psi_{j}$ for each index $k$ hence $\phi_{k}\in M_{s}^{1}%
(\mathbb{R}^{n})$ for every $k$ since $M_{s}^{1}(\mathbb{R}^{n})$ is a vector space.
\end{proof}

\begin{remark}
The proof does not trivially extend to the general case where the index set
$J$ is infinite because of convergence problems. The question whether the
Corollary extends to the general case is open.
\end{remark}

As a bonus we obtain the following new result about convex sums of Wigner distributions.

\begin{corollary}
\label{Prop5}Let $(\psi_{j})_{j}$ be a sequence in $M_{s}^{1}(\mathbb{R}^{n}%
)$. Let $(\phi_{j})_{j}$ be a sequence of functions in $L^{2}(\mathbb{R}^{n})$
and sequences $(\alpha_{j})$ and $(\beta_{j})$ of positive numbers such that
$\sum_{j}\alpha_{j}=\sum_{j}\beta_{j}=1$. We assume that $||\psi_{j}%
||=||\phi_{j}||=1$ for all $j$. If we have%
\[
\sum_{j}\alpha_{j}W\psi_{j}=\sum_{j}\beta_{j}W\phi_{j}%
\]
then $\phi_{k}\in M_{s}^{1}(\mathbb{R}^{n})$ for every $k$.
\end{corollary}

\begin{proof}
Both series are absolutely convergent in view of (\ref{bound}); for instance
\[
\sum_{j}\alpha_{j}|W\psi_{j}|\leq\left(  \tfrac{2}{\pi\hbar}\right)  ^{n}%
\sum_{j}\alpha_{j}=\left(  \tfrac{2}{\pi\hbar}\right)  ^{n}~.
\]
The function%
\[
\rho=(2\pi\hbar)^{n}\sum_{j}a_{j}W\psi_{j}=(2\pi\hbar)^{n}\sum_{j}\beta
_{j}W\phi_{j}%
\]
is the Wigner distribution of a density matrix generated by the Feichtinger
state $\{(\psi_{j},\alpha_{j})\}$. In view of Corollary \ref{Prop4} we must
then have $\phi_{j}\in M_{s}^{1}(\mathbb{R}^{n})$ for every $j$.
\end{proof}

\section{Discussion}

We have seen that the class of modulation spaces $M_{s}^{1}(\mathbb{R}^{n})$
provide us with a very convenient framework for the study of the Wigner
distribution of density matrices; it is far less restrictive than the
conventional use of the Schwartz space $\mathcal{S}(\mathbb{R}^{n})$ which
requires that the functions and all their derivatives be zero at infinity. In
addition, the topology of modulation spaces is simpler since they are defined
by a norm making them to Banach spaces, while that of $\mathcal{S}%
(\mathbb{R}^{n})$ is defined by a family of semi-norms making it to a
Fr\'{e}chet space. Another particularly attractive feature of modulation
spaces is that they allow to introduce an useful class of Banach Gelfand
triples (see \cite{cofelu08}). For instance,
\[
(S_{0}(\mathbb{R}^{n}),L^{2}(\mathbb{R}^{n}),S_{0}^{\prime}(\mathbb{R}^{n}))
\]
where $S_{0}^{\prime}(\mathbb{R}^{n})$ is the dual space of the Feichtinger
algebra $S_{0}(\mathbb{R}^{n})$ is such a triple. $S_{0}^{\prime}%
(\mathbb{R}^{n})$ consists of all $\psi\in S^{\prime}(\mathbb{R}^{n})$ such
that $W(\psi,\phi)\in L^{\infty}(\mathbb{R}^{2n})$ for one (and hence all)
$\phi\in S_{0}(\mathbb{R}^{n})$; the duality bracket is simply given by the
pairing
\begin{equation}
(\psi,\psi^{\prime})=\int_{\mathbb{R}^{2n}}W(\psi,\phi)(z)\overline
{W(\psi^{\prime},\phi)(z)}dz~. \label{dual151}%
\end{equation}
Since $S_{0}(\mathbb{R}^{n})$ is the smallest Banach space isometrically
invariant under the action of the metaplectic group its dual is essentially
the largest space of distributions with this property. The use of such triples
makes the use of the Dirac bra-ket notation much more natural and rigorous.
For instance, objects like $\langle\psi|\phi\rangle$ automatically have a
meaning for all $\phi\in S_{0}(\mathbb{R}^{n})$ and all $\psi\in S_{0}%
^{\prime}(\mathbb{R}^{n})$.

\begin{acknowledgement}
Maurice de Gosson has been supported by the grant P 33447\ of the Austrian
Research Foundation FWF. The authors thank Hans Feichtinger for illuminating
discussions about modulation spaces, and for having suggested various
improvements of the original text.
\end{acknowledgement}

\textbf{Data Availability. }All the data used for the present paper (here:
LaTex file) is freely available. No other data has been used or created.

\end{document}